\documentclass[11pt]{article}

\usepackage[utf8]{inputenc}
\usepackage[T1]{fontenc}
\usepackage[margin=1in]{geometry}
\usepackage{amsmath,amssymb,amsthm}
\usepackage{graphicx}
\usepackage{booktabs}
\usepackage{enumitem}
\usepackage{microtype}
\usepackage[hidelinks]{hyperref}

\theoremstyle{plain}
\newtheorem{theorem}{Theorem}
\newtheorem{corollary}{Corollary}[theorem]
\newtheorem{proposition}[corollary]{Proposition}

\title{Capability Advertisement as a Market for Lemons:\\
A Trust Layer for Heterogeneous Agent Networks}
\author{Gaurav Naresh Mittal \\ \texttt{gaurav.n.mittal@gmail.com}}
\date{\today}

\begin{document}
\maketitle

\begin{abstract}
Large language model (LLM) agents have begun to delegate work to one another. Protocols such as the
Model Context Protocol (MCP) and the Agent2Agent protocol (A2A) let an agent publish what it can do
and let other agents call it, and public registries of such agents are already appearing. These
protocols share an assumption inherited from a friendlier era of distributed computing: that an
advertised capability is a static, truthful fact. A real agent is none of these things. Its
competence is probabilistic, it varies with the input, it drifts when the underlying model is
updated, and---because the agent is itself a language model---it can describe itself with complete
confidence and be wrong. A caller therefore sees what an agent \emph{claims} to do, not what it
\emph{can} do, and has no principled way to tell a reliable provider from a fluent impostor.

We argue that these difficulties share a single, well-understood cause: the market for lemons. When
quality is hidden and claims are cheap, good and bad providers become indistinguishable, honest
reliability goes unrewarded, and the market decays toward its worst participants. The same forces
operate in an open agent network. Read this way, the cure is also known. Economics offers three
remedies---signaling, screening, and reputation---and none of them are present in today's agent
protocols.

We make four contributions. First, we give a failure taxonomy for heterogeneous agent networks that
names \emph{confident-wrong}---a non-adversarial, correlated subclass of Byzantine faults that
classical fault-tolerance mismodels---as the characteristic failure of this setting. Second, we model
capability advertisement as a lemons market and argue that faith-based protocols admit only a
low-trust equilibrium. Third, we propose the \textbf{Trust Layer}: a thin, protocol-agnostic layer
that adds probabilistic capability descriptors, screening, and reputation as a narrow waist above MCP
and A2A, and we show analytically that it admits a separating equilibrium under a stated condition.
Fourth, we give a reliability-composition bound for delegation chains and invoke the end-to-end
argument to delimit what such a layer can and cannot promise. The design needs no model retraining
and degrades gracefully when its trust anchors are absent or corrupt.

\end{abstract}

\section{Introduction}

Consider a planning agent assembling a contract review. It runs on a capable model, it has a budget,
and it would rather not do everything itself, so it decomposes the job and looks for help. One
subtask---\emph{summarize the indemnification clauses}---is delegated to another agent it found in a
registry. That agent advertises exactly this skill. It accepts the task, works for a few seconds, and
returns a clean, well-organized summary. The planner folds the result into the final document and
moves on.

The summary is wrong. The delegated agent runs on a smaller, cheaper model that produces fluent legal
prose without reliably tracking which party indemnifies which. Nothing in the exchange revealed this.
There was no exception, no timeout, no malformed response---only a confident answer that happened to
be false. The planner had no signal that anything had gone wrong, because the protocol gave it no
channel in which such a signal could even be expressed.

This is not an exotic edge case. As agents increasingly call agents they did not build and cannot
inspect, the exchange above becomes the common case rather than the cautionary one. Public directories
of MCP servers and A2A agents are multiplying, and the agents behind them sit on a heterogeneous mix of
model families, fine-tunes, and scaffolds whose reliability profiles differ enormously and are not
observable from the outside.

\paragraph{Why the obvious fixes do not work.}
The natural reaction is to ask the agent how good it is. But self-assessment is the problem, not the
solution: language models are poorly calibrated about their own competence, and an open network gives
every provider an incentive to describe itself generously. A second reaction is to test the agent,
which helps but is expensive to do well and says nothing about tomorrow, because a provider can
silently swap the model behind a stable advertisement, so a system can quietly degrade from one week to
the next. A third reaction is to treat a bad answer as a fault and route around it, but classical
failure models have no category for an answer that is live, responsive, protocol-compliant, and wrong
with high fluency.

\paragraph{The insight.}
Each of these difficulties is a symptom of one underlying condition. The caller observes an advertised
capability but not the true one; the two can diverge; and the provider, not the caller, knows which is
which. This is asymmetric information, and its canonical analysis is Akerlof's market for lemons
\cite{akerlof1970}. In a lemons market, buyers who cannot distinguish quality will pay only for the
average, good sellers withdraw because the average underprices them, the average falls, and the market
unravels toward its worst goods. An open agent network has the same structure: a provider who invests
in genuine reliability emits an advertisement indistinguishable from an overclaimer's, earns no premium
for the investment, and is selected against. The equilibrium is a market of lemons.

Naming the problem this way is useful because it imports a solution. The economics of asymmetric
information did not stop at the diagnosis; it produced three remedies---\textbf{signaling}
\cite{spence1973}, \textbf{screening} \cite{rothschild1976}, and \textbf{reputation} in repeated
interaction \cite{resnick2000}. Strikingly, none of these mechanisms exist in current agent protocols,
and, by the end-to-end argument \cite{saltzer1984}, none of them belong inside the wire protocol; they
belong in a thin layer above it. That layer is the subject of this paper.

\paragraph{Contributions.}
\begin{enumerate}[leftmargin=1.4em]
\item \textbf{A failure taxonomy} for heterogeneous agent networks that isolates \emph{confident-wrong}
(capability overclaim)---a non-adversarial, correlated subclass of Byzantine faults that classical
fault-tolerance mismodels---and maps each fault to an observable signal (Section~3).
\item \textbf{A market-for-lemons model} of capability advertisement, with an argument that faith-based
protocols---those that accept advertised capability at face value---admit only a low-trust, pooling
equilibrium (Sections~3 and~4.5).
\item \textbf{The Trust Layer}, a protocol-agnostic design that adds probabilistic capability
descriptors, screening through challenges and third-party attestation, and reputation with drift
detection, composed as a narrow waist above MCP and A2A. We argue it admits a separating equilibrium
when the cost of screening exceeds the gain from overclaiming, and that it degrades gracefully when its
trust anchors are weak or absent (Section~4).
\item \textbf{A reliability-composition bound} for delegation chains and an end-to-end placement
argument for what the layer should and should not promise (Section~4.6).
\end{enumerate}

The contribution is conceptual rather than experimental: we recognize an isomorphism and adapt
mechanisms with known properties. Section~6 is therefore an analysis---equilibrium arguments,
composition bounds, and grounding in published benchmarks---rather than a report of new measurements,
and we are explicit about the difference. The remainder of the paper covers background and the
cooperative-capability assumption (Section~2), the lemons formulation (Section~3), the design
(Section~4), a mapping onto existing protocols (Section~5), analysis (Section~6), related work
(Section~7), discussion (Section~8), and conclusions (Section~9).

\section{Agent Protocols and the Cooperative-Capability Assumption}

Today's agent protocols describe capability as a schema. In A2A, an agent publishes an \emph{agent
card} listing skills; in MCP, a server exposes \emph{tools}, each with a name, a description, and a
typed input/output schema; OpenAI-style function calling does much the same. These designs are good at
what they were built for: telling a caller how to \emph{invoke} a capability---what arguments to pass,
what shape of result to expect. They say nothing about whether the call will \emph{succeed}.

The omission is structural. A tool schema is a boolean assertion: the capability is present or it is
not. There is no field for how often the tool returns a correct result, no way to say that accuracy
falls off past a certain input length or outside a certain domain, no record of what evidence backs the
claim, and no expiry date after which the claim should be re-checked. A reliable agent and a confident
overclaimer therefore publish the \emph{same} advertisement. Table~\ref{tab:ads} summarizes the gap.

\begin{table}[t]
\centering
\caption{What current advertisements express ($\bullet$ present, $\circ$ absent).}
\label{tab:ads}
\begin{tabular}{lcccccc}
\toprule
Protocol & Invocation & Reliability & Context- & Provenance & Freshness & Trust \\
 & schema & & dependence & & & signal \\
\midrule
MCP tools        & $\bullet$ & $\circ$ & $\circ$ & $\circ$ & $\circ$ & $\circ$ \\
A2A agent cards  & $\bullet$ & $\circ$ & $\circ$ & $\circ$ & $\circ$ & $\circ$ \\
Function calling & $\bullet$ & $\circ$ & $\circ$ & $\circ$ & $\circ$ & $\circ$ \\
\bottomrule
\end{tabular}
\end{table}

\paragraph{System model.}
We consider an open network in which any party may publish an agent and any party may call one. An
\emph{agent} is a backend (a model, possibly fine-tuned) wrapped in a scaffold (tools, prompts, control
logic). Agents delegate tasks to one another, forming a directed \emph{delegation graph} whose depth is
not bounded in advance. A caller cannot inspect a callee's backend or scaffold; it sees only the
advertisement and the responses. We assume interactions repeat and that agent identities persist long
enough to accumulate history---assumptions we return to in Section~6.5, because the design leans on
them.

\paragraph{Why the classical failure model is not enough.}
Dependable-systems theory classifies faults as crash, omission, timing, or Byzantine
\cite{avizienis2004,lamport1982}. A confident-wrong answer is, formally, a Byzantine fault: arbitrary
incorrect output is exactly what that model admits. The trouble is that the standard Byzantine
\emph{response} fits this setting poorly. Byzantine fault tolerance assumes worst-case, adversarial
behavior and defends against it with replication that tolerates a bounded fraction of failures, on the
assumption that those failures are independent. Confident-wrong faults honor neither premise. They are
not adversarial; they are faithful reports of mistaken beliefs. And they are far from independent:
agents that share training data and prompts tend to fail on the same inputs in the same way, so a
supermajority can be confidently wrong together and voting buys little. The binary honest-or-faulty
abstraction also discards the quantity that matters most in practice, which is how likely a given answer
is to be correct. We therefore treat confident-wrong not as a fault outside the Byzantine class but as
the subclass this setting is made of---non-adversarial, correlated, and probabilistic---and one that the
machinery built for the general case mismodels. Section~3 gives it a name and a place among the other
faults.

\section{The Lemons Market in Capability Advertisement}

\subsection{An example that compounds}

Return to the contract-review planner, but let it delegate down a chain rather than to a single agent.
The planner asks Agent A to extract obligations; A delegates clause summarization to B; B calls C to
resolve cross-references. Each agent advertises the relevant skill, and each returns a confident,
plausible result. No hop reports trouble. Yet the probability that the \emph{final} answer is correct
is the product of the per-hop reliabilities, and a chain of three agents that are each right 85\% of the
time is correct only about 61\% of the time. The error is invisible at every step and severe end to end.
Heterogeneity makes this worse, not better: the agents fail in different, sometimes correlated ways, so
the planner cannot even assume that a second opinion is independent.

\subsection{The problem, stated}

Let an agent $i$ have a true reliability $r_i \in [0,1]$ on a task class---the probability that its
result is correct---possibly a function $r_i(x)$ of input features $x$ such as length or domain. Let
$a_i$ be the agent's \emph{advertised} capability. In current protocols $a_i$ is a bit (the skill is
listed or not), and crucially the caller observes $a_i$, never $r_i$. Providers may set $a_i$
strategically; nothing prevents an agent of low $r_i$ from advertising the same skill as one of high
$r_i$.

This is the lemons setup. Because the caller cannot observe $r_i$, it can offer only an average level of
trust to everyone advertising a given skill. That average underrewards a high-$r_i$ provider, who must
pay---in compute, in careful scaffolding, in evaluation---to \emph{be} reliable while earning no more
than an overclaimer who pays nothing. The high-reliability provider's rational move is to stop paying.
As reliable providers withdraw or stop investing, the average falls, and the network selects for
confident overclaimers. Absent a mechanism that ties trust to truth, the faith-based protocol has a
single stable outcome, and it is the bad one.

\subsection{A failure taxonomy}

To reason about mechanisms we first need vocabulary for what can go wrong. We distinguish:
\begin{itemize}[leftmargin=1.4em]
\item \textbf{Crash}---the agent does not respond.
\item \textbf{Timeout / omission}---the agent is too slow or drops the request.
\item \textbf{Refusal}---the agent declines (policy, safety, or capability awareness). \emph{Honest and
useful.}
\item \textbf{Partial / degraded}---the agent returns an incomplete or low-quality result and, ideally,
says so.
\item \textbf{Confident-wrong (capability overclaim)}---the agent returns a fluent, well-formed,
incorrect result and signals no doubt. \emph{This is the fault current protocols cannot see.}
\item \textbf{Drift}---the agent's true reliability changes (typically a backend update) while its
advertisement does not.
\item \textbf{Misroute}---the task reaches an agent whose advertised skill does not match the real need.
\end{itemize}

The first two are observable directly. Refusal and partial results are observable when the agent is
built to report them. Confident-wrong is the hard case: it leaves no native signal and surfaces only
through screening, a judge, or a post-condition check, which is why it is the fault a trust mechanism
must be designed around.

This taxonomy scopes the rest of the paper. Crash, timeout, and omission are already handled by ordinary
timeouts and retries; refusal and partial results help the caller whenever agents report them honestly.
What current protocols cannot see---confident-wrong, the drift that produces it silently, and the
misrouting that boolean skill claims invite---is what the Trust Layer of Section~4 is built for:
descriptors and screening address confident-wrong and misrouting, while freshness and drift detection
address drift.

\section{A Trust Layer for Heterogeneous Delegation}

\subsection{Overview}

Our design adds a thin layer between a caller and the transport protocol it uses to reach a callee. It
does not replace MCP or A2A; it sits above them as a \emph{narrow waist}, in the sense of the Internet
hourglass \cite{clark1988}: a minimal common interface that many backends below and many planners above
can share. The layer has three parts, each corresponding to one classical remedy for asymmetric
information.
\begin{itemize}[leftmargin=1.4em]
\item \textbf{Probabilistic capability descriptors} let a good provider \emph{signal} its quality
credibly.
\item \textbf{Screening} lets a caller \emph{separate} providers by quality before trusting them.
\item \textbf{Reputation with drift detection} lets the network \emph{remember}, so that today's claim
is disciplined by yesterday's results and re-checked when the world changes.
\end{itemize}

Four principles guide the design: it should be agnostic to the model and to the transport; it should be
deployable incrementally, working alongside agents that ignore it; and it should respect the end-to-end
argument by promising only what a protocol layer can honestly promise.

\subsection{Probabilistic capability descriptors (signaling)}

The descriptor replaces the boolean claim with a structured one. Instead of ``I summarize
indemnification clauses,'' an agent publishes that it does so with reliability approximately 0.91,
calibrated on a named benchmark of a stated size and date, with accuracy that falls past roughly eight
thousand tokens, produced by backend version v3.2, valid for thirty days. It is, in effect, a nutrition
label for competence.

The point is not merely richer metadata; it is \emph{credible} metadata. A bare number is cheap to
inflate, so the descriptor's load-bearing field is \textbf{provenance}: the evidence behind the
claim---which evaluation, when, on how many samples. Provenance is what turns the descriptor from cheap
talk into a signal in Spence's sense \cite{spence1973}, because a claim that points to reproducible
evidence is one an overclaimer cannot cheaply imitate. We considered and rejected two extremes: the bare
boolean schema of today, which carries no signal, and heavyweight semantic-service ontologies of the
2000s (OWL-S and kin), which demanded so much formality that they were never adopted at scale. The
descriptor aims for the smallest schema that can carry a credible signal.

\subsection{Screening: challenges and attestation}

Signals help, but a careful caller will also \emph{test}. Screening \cite{rothschild1976} is the
caller-side complement to provider-side signaling, and it takes two forms here. The first is a
\textbf{challenge}, or canary: before entrusting an agent with the real contract, the caller sends a
task whose answer it already knows and checks the response. The second is \textbf{attestation}: an
independent party evaluates the agent and issues a signed statement of the result, which the agent
presents like a credential. Attestation matters because challenges are expensive to design well and
wasteful to repeat at every caller; a shared attestation lets one party's careful evaluation serve many.

What screening adds beyond information is \emph{cost asymmetry}. Passing a real challenge, or earning a
genuine attestation, costs more for an agent that lacks the capability than for one that has it. That
cost gap is what makes the two types separate rather than pool---the formal point we take up in
Section~4.5.

\subsubsection{The attestation trust model}

Attestation invites an obvious and fatal-looking objection: does it not just relocate the trust problem
to the attester? If the caller must trust a capability authority, how does anyone know the authority is
neutral and has not been bought? Web public-key infrastructure shows the danger is real---certificate
authorities have been compromised, and the browser trust-root model failed quietly when they were. We
take the objection seriously, and our answer has three moves.

\emph{First, prefer verifiability over authority.} An attestation should not say ``trust me, this agent
is good.'' It should say ``I ran this public benchmark with this seed; here is the signed, reproducible
transcript.'' That turns the attester from a judge into a notary whose work can be re-run or
spot-checked. A notary that fabricates results is caught by anyone who reproduces them. The precedent is
Certificate Transparency \cite{laurie2013}, which addressed untrustworthy certificate authorities not by
making them honest but by logging their issuances publicly so that misbehavior becomes detectable after
the fact. Attestations live in the same kind of append-only log.

\emph{Second, treat attestation as an optimization, not a foundation.} It is one of three screening
tools, beside the agent's own signed descriptor and the caller's own challenges. Its only job is to
amortize evaluation cost across many callers. A caller who trusts no attester does not lose the ability
to delegate safely; it falls back to screening agents itself and learning from its own results---slower
and costlier, but intact. A corrupt attester therefore does not break the network; it removes an
efficiency and reverts the caller to the no-attester baseline. The system degrades gracefully rather
than catastrophically, which is a far stronger guarantee than ``the authority is honest.''

\emph{Third, make corruption costly and self-correcting.} An attester can post a bond against its
attestations, to be slashed to a harmed caller if an audit re-run shows inflation; callers can run
several attesters in parallel and weight them; and, decisively, callers compare an attester's
predictions against their own downstream outcomes and discount attesters that systematically overpredict.
This recursion bottoms out: an attester is judged, in the end, by real task results the caller observes
directly, not by another authority.

We do not claim this makes attestation incorruptible. Collusion between a provider and an attester still
pays in the short run, before the caller has gathered enough of its own outcomes to notice; re-running
evaluations costs money and a transparency log catches fabrication but not subtle bias in test selection;
and some capabilities are not cheaply reproducible, in which case attestation reverts toward bare
authority and the trust problem genuinely returns. We treat these as scope boundaries (Section~6.5)
rather than as solved.

\subsection{Reputation and drift (the repeated game)}

The third remedy is memory. Every real outcome---corroborated by a downstream check, a challenge, or a
later contradiction---updates the agent's reputation, and reputation propagates across the delegation
graph in the spirit of EigenTrust \cite{kamvar2003}, so that an agent vouched for by trusted agents
inherits some standing. Memory converts a one-shot bluff into a losing long-run strategy, the familiar
discipline of repeated games \cite{resnick2000}.

Two features are specific to language-model agents. \textbf{Freshness}: descriptors and attestations
carry a time-to-live, after which they are discounted or re-checked, so that stale claims do not accrue
unearned trust. \textbf{Drift detection}: when a provider changes the backend behind a stable
advertisement, the layer treats the version change as an event that expires prior evidence and forces
re-validation---directly addressing the ``it got worse overnight'' failure that no current protocol can
notice. Optionally, an agent may \emph{stake} reputation or budget on a claim, to be slashed on a
confirmed confident-wrong result, sharpening the incentive to advertise honestly.

\subsection{Why the layer changes the equilibrium}

The mechanisms are individually familiar; the claim worth making is about their joint effect. We sketch
it informally here; Appendix~A gives the formal model, the two theorems, and their proofs.

Under a faith-based protocol, advertised capability is free and unverified---a cheap-talk setting
\cite{crawford1982}---so a low-reliability provider can mimic a high-reliability one at no cost. No
advertisement separates the types, the caller can do no better than treat them alike, and investment in
reliability earns no return. The only stable outcome is the pooling equilibrium of Section~3.2---the
lemons market.

Now add the layer. Signaling attaches a verifiable cost to a credible claim; screening attaches a cost
to \emph{passing} as capable; reputation attaches a future cost to being caught wrong. Write $g$ for the
one-shot gain an agent expects from overclaiming and $c$ for the combined cost it must bear to sustain
the overclaim through screening and across repeated interactions. When $c > g$, overclaiming is no
longer profitable, honest advertising becomes a best response, and a \emph{separating} equilibrium
exists in which high- and low-reliability providers send distinguishable signals and callers route
accordingly. The design's job is to push $c$ above $g$---cheaply enough that honest providers will bear
it, steeply enough that overclaimers will not. The condition also delimits the design's reach: where
credible challenges cannot be built or identities do not persist, $c$ cannot be raised, and the lemons
problem in that corner may be genuinely unsolvable at the protocol layer.

\subsection{Composing reliability, and what the layer should promise}

Because tasks flow down chains, a caller needs to reason not about one hop but about the whole path. If
hop $k$ succeeds with probability $r_k$ and failures are independent, end-to-end reliability is
$\prod_k r_k$, the compounding effect of Section~3.1; where the layer adds a verification step of
catch-rate $v_k$ at a hop, the residual undetected-error rate falls accordingly, and a caller can bound
end-to-end reliability rather than merely hope for it. Independence is an idealization---heterogeneous
agents share training data and failure modes---so the bound should be read as a planning tool whose
assumptions are stated, not as a guarantee.

This is also where the end-to-end argument \cite{saltzer1984} disciplines our ambitions. The layer can
carry signals, run screening, and keep reputation; it can \emph{bound and price} reliability and route to
improve it. It cannot \emph{ensure} that a particular answer is semantically correct, because only the
endpoints---ultimately the application and its human owners---hold the ground truth. A protocol that
promised more would itself be overclaiming. We therefore design the layer to make correctness checkable
and accountable rather than guaranteed.

\section{Realizing the Trust Layer over MCP and A2A}

The design is built to fit the protocols that already exist. A descriptor maps onto an A2A agent card as
additional fields beside each advertised skill, and onto an MCP tool as annotations beside the tool
schema. The added fields are optional and additive: an agent that ignores them still interoperates, and
a caller that finds them absent falls back to a default, cautious trust level and to its own screening.
A small gateway can sit in front of a heterogeneous fleet, speak each transport, and normalize claims
into the common descriptor form---the narrow waist in practice. Attestations and reputation records are
published to an append-only log addressed by agent identity, so that any caller can fetch an agent's
history without a central clearinghouse mediating each lookup. None of this requires changing a model; it
changes only what travels alongside the call. We leave a reference schema and validator to future work.

\section{Analysis}

We report no measurements of a live system, and we will not dress the analysis up as one. What we offer
is threefold: the equilibrium results of Appendix~A; an \emph{illustrative} agent-based simulation that
visualizes those results (Section~6.3); and grounding in reliability measurements others have already
published. The simulation is a closed world---no language model is queried, every provider's reliability
is a stipulated parameter, and the correctness oracle is assumed---so it can show only that the mechanism
behaves as the proofs say, never that the magnitudes match any real deployment. Building and measuring
such a deployment is the natural next step and is out of scope here.

\subsection{Setup}

The arguments below assume the system model of Section~2: an open network, repeated interactions,
persistent identities, and the existence of at least some tasks whose answers can be checked. Where those
assumptions fail, the conclusions weaken in ways we flag.

\subsection{Research questions}

\begin{itemize}[leftmargin=1.4em]
\item \textbf{RQ1.} Does a faith-based protocol collapse to a low-trust equilibrium? We argue yes
(Section~3.2, Theorem~1): with free, unverified advertisement, no signal separates the types and pooling
is the only stable outcome. The simulation (Figure~\ref{fig:collapse}) shows the collapse directly.
\item \textbf{RQ2.} Does the Trust Layer admit a separating equilibrium, and at what cost? We argue yes
when $c > g$ (Theorem~2): when the cost of sustaining an overclaim exceeds its gain, honest advertising
is a best response. The simulation locates the transition precisely at $c = g$
(Figure~\ref{fig:threshold}).
\item \textbf{RQ3.} How does reliability scale with chain depth, with and without the layer? The
composition of Section~4.6 makes the dependence explicit---naive chains decay multiplicatively, while a
verification step at each hop curbs the decay---so the layer's value grows with depth
(Figure~\ref{fig:depth}).
\item \textbf{RQ4.} Is confident-wrong prevalent enough to matter? Published benchmarks say yes. Holistic
evaluations report wide reliability spreads across model families \cite{liang2022}, agent benchmarks show
large gaps on multi-step tasks \cite{liu2023}, and tool-use studies document frequent, confidently-made
errors \cite{qin2023}. The fault we target is the common case, not a corner case.
\item \textbf{RQ5.} What does screening cost? In the worst case, one extra challenge call per new
provider-task pairing; attestation amortizes this across callers, and reputation lets a caller skip
screening for providers it already trusts, so steady-state overhead falls well below the first-encounter
cost.
\end{itemize}

\subsection{An illustrative simulation}

To make the equilibrium results tangible we built a small agent-based simulation (a few hundred lines of
Python, available at \url{https://github.com/grvnmttl/lemons-market-sim}). A population of providers each
choose whether to invest in genuine reliability---at an idiosyncratic cost $\kappa_i$---and what to
advertise, then revise by myopic best response over many rounds, under two regimes (faith-based and Trust
Layer), averaged over 24 random seeds. The parameters are stylized: low and high true reliabilities
$r_L = 0.55$ and $r_H = 0.92$, so the overclaim gain is $g = 0.37$. As stressed in the setup, this
\emph{illustrates} Appendix~A; it does not test the design against reality.

Figure~\ref{fig:collapse} shows the central contrast (RQ1, RQ2). Both regimes begin from a healthy market
in which half the providers have invested. Under faith-based advertising the market unravels within a few
rounds: realized reliability falls to $r_L$ and the share investing in genuine reliability collapses to
near zero---the lemons outcome of Theorem~1 and Corollary~1.1. Under the Trust Layer with $c = 1.5\,g$,
reliability holds near $r_H$ and the invested share settles at $0.58$, essentially the predicted
$P(\kappa < g - c_H) = 0.583$ of Corollary~2.1.

\begin{figure}[t]
\centering
\includegraphics[width=\linewidth]{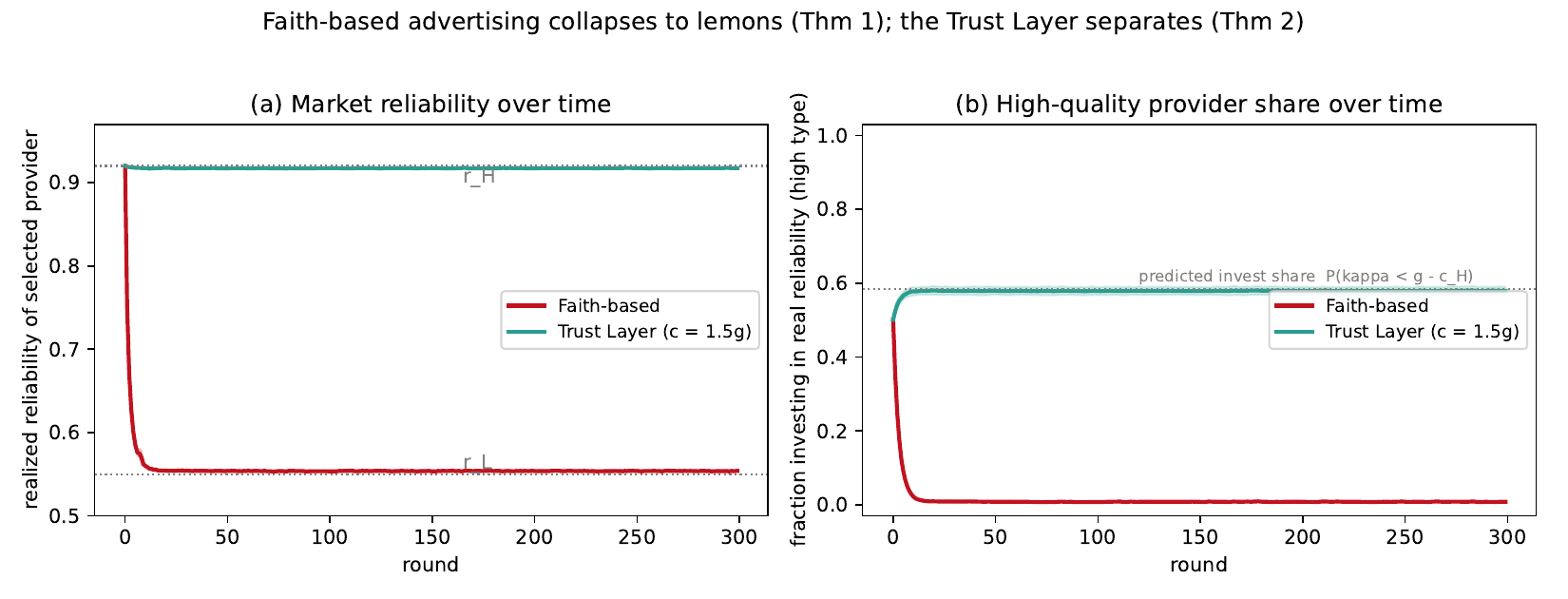}
\caption{Faith-based advertising versus the Trust Layer over time (mean of 24 seeds; bands are $\pm 1$
s.d.). \textbf{(a)} The realized reliability obtained by a caller collapses to $r_L$ under faith-based
advertising but holds near $r_H$ under the layer. \textbf{(b)} The share of providers investing in
genuine reliability falls to near zero under faith-based advertising (Corollary~1.1) and settles at the
predicted $P(\kappa < g - c_H)$ under the layer (Corollary~2.1).}
\label{fig:collapse}
\end{figure}

Figure~\ref{fig:threshold} sweeps the screening cost (RQ2). As $c$ crosses $g$ the steady state flips
sharply: overclaiming falls from a majority of providers to near zero, and market reliability jumps from
about $0.65$ to $r_H$. The transition sits exactly at $c = g$, as Theorem~2 requires.

\begin{figure}[t]
\centering
\includegraphics[width=0.72\linewidth]{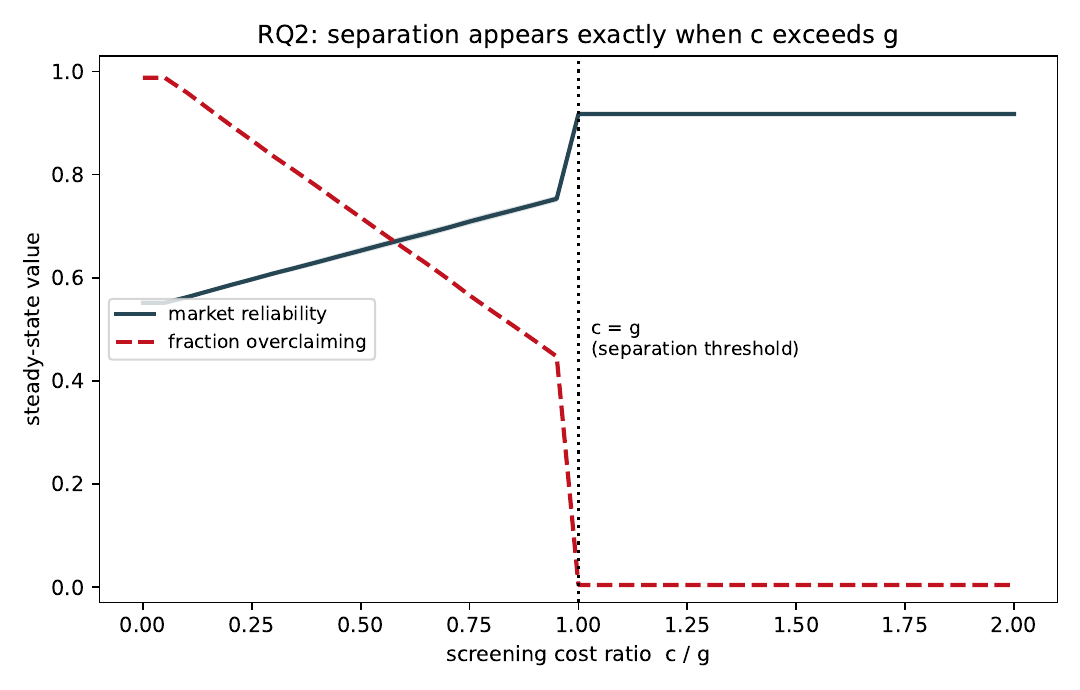}
\caption{Steady-state outcomes as the screening cost $c$ varies relative to the overclaim gain $g$.
Separation appears precisely at $c = g$ (Theorem~2): overclaiming collapses and market reliability jumps
to $r_H$.}
\label{fig:threshold}
\end{figure}

Figure~\ref{fig:depth} composes per-hop reliability into delegation chains (RQ3). Because reliability
multiplies along a chain, the faith-based market---reliable only about $0.55$ per hop---is correct barely
5\% of the time at depth five, whereas the Trust Layer with per-hop verification stays above $0.8$. The
layer's advantage widens with depth, which is the direction real agent networks are heading.

\begin{figure}[t]
\centering
\includegraphics[width=0.72\linewidth]{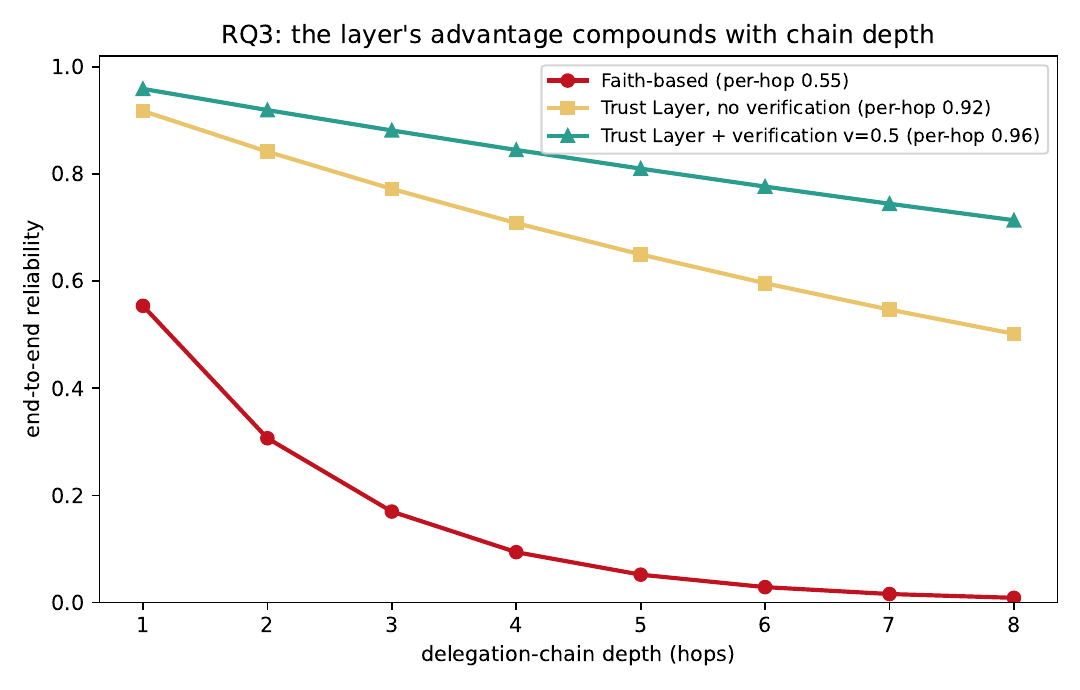}
\caption{End-to-end reliability versus delegation-chain depth, using each regime's steady-state per-hop
reliability. Multiplicative decay punishes the faith-based market severely; the Trust Layer with per-hop
verification (catch-rate $v = 0.5$) degrades far more gently.}
\label{fig:depth}
\end{figure}

\subsection{Ablation}

Remove descriptors and providers cannot signal, so screening must carry the entire burden and the
per-call cost rises. Remove screening and signals revert to cheap talk, collapsing toward pooling. Remove
reputation and each interaction stands alone, so drift goes unnoticed and a provider can overclaim afresh
against every new caller. Each part earns its place by guarding a failure the others do not.

\subsection{Limitations and failure modes}

Several boundaries deserve to be stated plainly. \textbf{Calibration bootstrapping}: a provider's
self-reported reliability is only as good as the evaluation behind it, and a provider cannot fully
self-assess; the design leans on external evidence precisely because of this, but cold-start cases have
little evidence to lean on. \textbf{Attestation under collusion}: as Section~4.3.1 concedes, a captured
attester pays off in the short run, and verifiability catches fabrication but not biased test selection.
\textbf{Correlated errors}: heterogeneous agents are not independent, so cross-checking and voting are
weaker than the arithmetic suggests, and the composition bound should be read with that caveat.
\textbf{Reputation gaming and fairness}: Sybil identities, collusive vouching, and incumbency advantage
are real, and a reputation system can entrench early movers; mitigations exist but none is complete.
\textbf{Non-reproducible capabilities}: where outputs are stochastic or test sets are private, screening
and attestation lose force and the trust problem returns in full. \textbf{Closed-world simulation}: the
Section~6.3 simulation illustrates only the model's own dynamics; its magnitudes are artifacts of
stipulated parameters and carry no empirical weight about real systems.

\section{Related Work}

\paragraph{Agent protocols.}
MCP \cite{anthropic2024} and A2A \cite{google2025}, along with function-calling interfaces and frameworks
such as AutoGen and LangGraph, define how agents advertise and invoke capabilities. They are the substrate
this work builds on, and the capability-and-trust gap they leave is the gap we fill; we add a layer above
them rather than competing with them.

\paragraph{Trust and verification for LLM agents.}
Trust in LLM-agent systems has become an active research area. Surveys now catalog threats and
countermeasures across single- and multi-agent settings \cite{yu2025}, and several systems attach
per-agent credibility or reputation scores so that agents can discount unreliable peers \cite{ebrahimi2025}.
Security analyses of MCP document \emph{semi-honest} servers that obey the wire protocol while behaving
differently from what they advertise \cite{hou2025}---the adversarial counterpart of the confident-wrong
fault we study. The verification tools these systems lean on, LLM-as-judge foremost among them, are
themselves of uneven reliability \cite{gu2024}. We differ from this work in what we provide more than in
what we observe. Where it treats untrustworthiness as a security threat, or supplies a trust score whose
guarantees stay implicit, we give an economic account: capability advertisement is a market for lemons,
which is why the problem is structural rather than incidental, and why a mechanism helps only once
screening costs more than overclaiming gains (Section~4.5). Our overclaiming is also non-adversarial---a
faithful report of a wrong self-belief---so it survives among the cooperative agents that a security
framing assumes away. And we site the remedy in a protocol-layer narrow waist above MCP and A2A, not
inside one framework.

\paragraph{Classical multi-agent systems.}
The problem of one agent delegating to another is old. The Contract Net Protocol \cite{smith1980}
formalized task announcement and bidding; FIPA standardized agent communication languages; the
belief--desire--intention model \cite{rao1995} and platforms like JADE \cite{bellifemine2007} built on
these. That tradition largely assumed cooperative agents whose advertised competence was truthful and
stable. Our departure is to drop that assumption: when competence is probabilistic and self-reported by a
language model, the classical machinery has no account of overclaiming, and that is precisely what the
Trust Layer supplies.

\paragraph{Trust and reputation.}
Reputation systems for online marketplaces and peer-to-peer networks \cite{resnick2000,kamvar2003} inform
our reputation sublayer. We build on them, adding freshness and drift detection for the specific problem
of a backend that changes beneath a stable advertisement.

\paragraph{Economics of asymmetric information.}
The market for lemons \cite{akerlof1970} is our framing, and signaling \cite{spence1973} and screening
\cite{rothschild1976} are the remedies we operationalize. The contribution is to recognize that capability
advertisement \emph{is} a lemons market and that these remedies map onto concrete protocol mechanisms.

\paragraph{Distributed systems.}
We draw on Byzantine fault tolerance \cite{lamport1982} and the dependability taxonomy \cite{avizienis2004}
to position confident-wrong as a distinct fault, on the end-to-end argument \cite{saltzer1984} to decide
what the layer should promise, and on the narrow-waist design of the Internet \cite{clark1988} for the
deployment story.

\paragraph{LLM evaluation.}
Holistic and agentic benchmarks \cite{liang2022,liu2023,qin2023} supply the empirical grounding that
confident-wrong is prevalent, and model cards \cite{mitchell2019} are the documentation lineage our
descriptors extend toward machine-readable, probabilistic claims.

\section{Discussion}

If capability advertisements become credible, a market can form around them. Reliability could be priced,
attesters could compete on the rigor of their evaluations, and agents could be selected on demonstrated
competence rather than on fluent self-description. That is the optimistic reading, and it is worth
pursuing, but it carries its own hazards: attestation can centralize, reputation can entrench incumbents,
and a market that prices reliability can also be gamed by those who learn to look reliable.

The framing generalizes beyond language models. Any open network of services with hidden quality and cheap
claims has a lemons problem, and the same triad---signal, screen, remember---applies. The specifics that
make our setting distinctive are that quality is probabilistic and context-dependent rather than binary,
and that it drifts on a timescale short enough to matter.

Several problems remain open. Cross-checking under correlated, heterogeneous errors is genuinely unsolved
and may deserve its own treatment. Decentralizing attestation without recreating a trusted root is hard.
And bootstrapping calibration where no evaluation history exists is a chicken-and-egg problem the design
mitigates but does not dissolve. We see these as the agenda this framing opens, not as gaps that undermine
it.

\section{Conclusion}

Today's agent protocols assume that an agent's advertised capability is true, static, and known.
Language-model agents violate all three, and the failures they produce have a single economic root: a
market for lemons, in which honest reliability cannot distinguish itself and the network drifts toward
confident overclaimers. Seen this way, the remedies are the ones economics has long known---signaling,
screening, and reputation---none of which today's protocols provide.

We have proposed the Trust Layer, a thin and protocol-agnostic narrow waist that supplies exactly these
three mechanisms above MCP and A2A, and we have argued that it turns a pooling equilibrium of lemons into
a separating equilibrium in which reliability is provable and rewarded, whenever the cost of sustaining an
overclaim can be pushed above its gain. The design needs no retraining, deploys incrementally, and---because
its attestation anchor is an optimization rather than a foundation---degrades gracefully when trust is
scarce. The most pressing next step is to build it and measure it: to learn, in a live heterogeneous
network, how high that cost must be set, and where the corners are in which the lemons cannot be cleared at
all.

\appendix

\section{Equilibrium Analysis}

This appendix formalizes the claim of Section~4.5: a faith-based protocol admits only a pooling
(low-trust) equilibrium, whereas the Trust Layer's screening admits a separating equilibrium exactly when
the cost of sustaining an overclaim exceeds the gain from it. The model is a standard signaling game in the
lineage of Akerlof, Spence, and Stiglitz. What is new is the mapping: casting capability advertisement as a
lemons market turns a vague worry about untrustworthy agents into two precise claims---that trust collapses
without a mechanism (Theorem~1), and that a mechanism restores it once screening costs more than
overclaiming gains (Theorem~2).

\subsection{Model}

There is a population of \emph{providers} and a population of \emph{callers}. Each provider has a hidden
\textbf{type} $\theta \in \{L, H\}$ with true task-success probabilities $r_L < r_H$ in $[0,1]$. A fraction
$\lambda \in (0,1)$ of providers are type $H$; write the prior-mean reliability as
\[
\rho_{\text{pool}} \;=\; \lambda\, r_H + (1-\lambda)\, r_L .
\]
A provider observes its own type and publishes an \textbf{advertisement} (a message) $m \in M$. A caller
observes $m$---and, where a screening technology exists, a verifiable credential---forms a posterior belief,
and decides whom to delegate to and what to pay.

We summarize the caller side by a single \textbf{trust price} $B(\rho)$: the benefit a provider receives
when callers perceive its success probability to be $\rho$ (more work routed to it, a higher fee, or both).
We assume:
\begin{itemize}[leftmargin=1.4em]
\item \textbf{(A1) Monotone trust price.} $B$ is strictly increasing in $\rho$. Write $B_H = B(r_H)$,
$B_L = B(r_L)$, $B_{\text{pool}} = B(\rho_{\text{pool}})$, so $B_L < B_{\text{pool}} < B_H$.
\item \textbf{(A2) Two types.} $r_L < r_H$. (The continuum case is Remark~A.5.3.)
\item \textbf{(A3) Cheap talk under faith-based advertising.} Messages are costless and type-independent:
any type can send any $m \in M$ at zero cost, and there is no verification.
\item \textbf{(A4) Single-crossing screening.} A screening technology lets a provider obtain a \emph{high
credential}. Obtaining it costs $c_H \ge 0$ for a true $H$ and $c_H + c$ for a true $L$, with $c > 0$. That
is, faking competence on a genuine challenge is strictly costlier than demonstrating it. (This is the
Spence--Mirrlees single-crossing property \cite{spence1973,mirrlees1971}; it presumes verifiable task
instances exist.)
\item \textbf{(A5) Rational play.} Providers are risk-neutral and maximize trust price minus cost; callers
are Bayesian and the trust price they pay is increasing in their posterior (by A1).
\end{itemize}

Define the \textbf{recognition premium}, equivalently the one-shot \textbf{overclaim gain},
\[
g \;:=\; B_H - B_L \;>\; 0 .
\]
This is what a type $L$ would capture by passing itself off as $H$.

\subsection{Equilibrium concepts}

The solution concept is \textbf{Perfect Bayesian Equilibrium (PBE)}: a strategy profile and a belief system
such that strategies are sequentially rational given beliefs, and beliefs are derived from strategies by
Bayes' rule on the equilibrium path. An equilibrium is \textbf{separating} if distinct types induce distinct
posteriors (the caller can tell them apart) and \textbf{pooling} if all participating types induce the same
posterior (the caller cannot). Where multiple PBE exist we invoke the \textbf{Cho--Kreps Intuitive
Criterion} \cite{cho1987} as a refinement (Remark~A.5.1).

\subsection{Faith-based advertising collapses to pooling}

\begin{theorem}
Under faith-based advertising (A3), no separating PBE exists. Every PBE is outcome-equivalent to a pooling
equilibrium in which advertising is uninformative and the perceived reliability of every participating
provider equals the population mean $\rho_{\text{pool}}$.
\end{theorem}

\begin{proof}
\emph{(No separation.)} Suppose, for contradiction, an informative PBE in which two messages induce
different posteriors. Let $\bar m$ be an on-path message inducing the highest posterior reliability
$\bar\rho$, and let $\underline m$ be an on-path message inducing a strictly lower posterior
$\underline\rho < \bar\rho$, sent by some type $t$. Beliefs attach to the \emph{message}, not to the
sender's identity, so any provider sending $\bar m$ is perceived as $\bar\rho$. Consider type $t$: its
equilibrium payoff is $B(\underline\rho)$, while deviating to $\bar m$ yields $B(\bar\rho) - 0 >
B(\underline\rho)$ by (A1) and the costlessness of messages (A3). The deviation is strictly profitable,
contradicting equilibrium. Hence all participating types induce a single posterior. \emph{(Value of the
pool.)} With one on-path posterior and all types participating, Bayes' rule gives $P(H \mid m) = \lambda$,
so the perceived reliability is $\lambda r_H + (1-\lambda) r_L = \rho_{\text{pool}}$.
\end{proof}

The economic content is that without verification a low type can always mimic a high type's words at no
cost, so words carry no information and trust defaults to the average---the market for lemons. Theorem~1 is
the costless-message (cheap-talk) special case in the lineage of Crawford and Sobel \cite{crawford1982}.

\begin{corollary}[Investment collapse]
Add a prior stage in which a provider may pay $\kappa > 0$ to become type $H$ (otherwise it is $L$), with
reliability rewarded only through perception. Under faith-based advertising, no provider invests; the
equilibrium fraction of high types is $\lambda = 0$ and realized reliability converges to $r_L$.
\end{corollary}

\begin{proof}
A single provider is of negligible mass and, by Theorem~1, cannot signal its type, so its perceived
reliability is $\rho_{\text{pool}}$ whether or not it invests. Its gross benefit is thus
$B(\rho_{\text{pool}})$ either way, while investing additionally costs $\kappa > 0$. Not investing strictly
dominates, so no one invests and the realized population is all $L$.
\end{proof}

This is the unraveling: because reliability cannot be demonstrated, paying to be reliable is strictly
irrational, and the network decays to its weakest type. Corollary~1.1 drives the collapse through the
\emph{investment} margin. The same force also operates on a \emph{participation} margin---high types exiting
the market rather than accept a pooled price---which we make explicit so that no exit channel is left
unmodelled.

\begin{proposition}[Participation and unraveling]
Suppose each type $\theta$ has an outside option worth $u_\theta$, with $u_L \le u_H$, and participates only
if its market payoff is at least $u_\theta$. Under faith-based advertising: (i) if
$B(\rho_{\text{pool}}) \ge u_H$, full participation is sustained and the market pools at $\rho_{\text{pool}}$;
(ii) if $B(\rho_{\text{pool}}) < u_H$, the high type exits, the surviving market contains only low types and
is perceived at $r_L$, and it trades iff $B(r_L) \ge u_L$.
\end{proposition}

\begin{proof}
By Theorem~1 any set of participants pools at the mean reliability of that set, and each participant earns
the corresponding trust price. (i) If both types participate the perceived value is $\rho_{\text{pool}}$ and
each earns $B(\rho_{\text{pool}})$; the high type's participation constraint $B(\rho_{\text{pool}}) \ge u_H$
holds by assumption, and the low type's holds a fortiori since $u_L \le u_H$ and $\rho_{\text{pool}} \ge
r_L$. (ii) If $B(\rho_{\text{pool}}) < u_H$ the high type's participation constraint fails under full
participation, so it exits; with only low types remaining the pooled mean is $r_L$, giving each survivor
$B(r_L)$, and the low type stays iff $B(r_L) \ge u_L$.
\end{proof}

With two types the exit is a single step; with a continuum of qualities it cascades, since each marginal
exit lowers the participant mean and can trigger the next, so the market can unravel all the way to the
lowest type---Akerlof's original result. Either way, the high-quality end of an unverified agent market is
driven out, whether by withdrawing effort (Corollary~1.1) or by withdrawing entirely (Proposition~1.2).

\subsection{Screening restores separation}

\begin{theorem}
Add the single-crossing screening technology (A4) with a binary credential. A separating PBE exists if and
only if
\[
c_H \le g \qquad\text{and}\qquad c \;>\; g - c_H ,
\]
and whenever it exists it is unique in form: the true-$H$ providers obtain the credential and are perceived
as $H$ ($\rho = r_H$), while $L$ providers do not and are perceived as $L$ ($\rho = r_L$). In the normalized
case $c_H = 0$ the condition reduces to $c > g$: separation is sustainable exactly when the cost of
sustaining an overclaim exceeds the gain from it.
\end{theorem}

\begin{proof}
Consider the candidate profile: type $H$ obtains the credential and is perceived as $r_H$; type $L$ does not
and is perceived as $r_L$; off the path, any provider lacking the credential is believed to be $L$. We check
sequential rationality.

\emph{Type $L$ (incentive not to mimic).} Honesty yields $B_L$ (perceived $L$, no screening cost). Mimicking
---obtaining the credential---yields $B_H - (c_H + c)$. Type $L$ prefers honesty iff
\[
B_L \;\ge\; B_H - c_H - c \;\Longleftrightarrow\; c \;\ge\; (B_H - B_L) - c_H \;=\; g - c_H,
\]
with strict preference iff $c > g - c_H$. \textbf{(IC-L)}

\emph{Type $H$ (incentive to screen).} Screening yields $B_H - c_H$; abstaining yields $B_L$ (perceived
$L$). Type $H$ screens iff
\[
B_H - c_H \;\ge\; B_L \;\Longleftrightarrow\; c_H \;\le\; g. \quad \textbf{(IC-H)}
\]

\emph{Caller.} On the path, the credential is held only by $H$, so the posterior $r_H$ given a credential and
$r_L$ otherwise is Bayes-consistent; paying a trust price increasing in the posterior is optimal by (A5).
Both incentive constraints hold precisely when $c_H \le g$ and $c > g - c_H$, so the profile is a PBE. This
proves sufficiency.

\emph{Necessity and uniqueness.} Consider any separating PBE. Cheap-talk messages cannot separate types (the
argument of Theorem~1), so separation must ride on the credential---the only verifiable instrument---and
exactly one type holds it, the holder perceived $r_H$ and the non-holder $r_L$. The holder must be the high
type: if instead $L$ held the credential and $H$ did not, then $H$---perceived $r_L$---could acquire it at
cost $c_H$ to be perceived $r_H$, a change of $g - c_H$; moreover for $L$ to hold it at all requires
$B_H - c_H - c \ge B_L$, i.e.\ $c \le g - c_H$, which combined with $c>0$ forces $c_H < g$ and hence
$g - c_H > 0$, so $H$'s deviation is strictly profitable and breaks the candidate. Thus any separating PBE
has $H$ holding and $L$ not. Given that profile, $L$'s incentive not to mimic requires $B_L \ge B_H - (c_H +
c)$, i.e.\ $c \ge g - c_H$ (strict for strict incentives), and $H$'s incentive to hold rather than pool down
requires $B_H - c_H \ge B_L$, i.e.\ $c_H \le g$. Both conditions are therefore necessary for \emph{any}
separating PBE, not merely for the constructed one, and they are sufficient by the construction above. Hence
they characterize separation as such, and the separating profile is unique. Setting $c_H = 0$ gives $c > g$.
\end{proof}

\begin{corollary}[Investment restored]
With screening and separation, in the prior investment stage a provider invests to become $H$ iff
$\kappa \le g - c_H$. For investment cheap relative to the recognition premium, providers invest and
equilibrium reliability rises toward $r_H$.
\end{corollary}

\begin{proof}
Investing yields the separating high payoff net of screening, $B_H - c_H$; not investing yields $B_L$.
Investing is optimal iff $B_H - c_H - \kappa \ge B_L$, i.e.\ $\kappa \le g - c_H$.
\end{proof}

Theorems~1 and~2 together are the formal version of the paper's thesis: the same network that collapses to
lemons under faith-based advertising (Theorem~1, Corollary~1.1) separates and rewards genuine reliability
once screening makes overclaiming cost more than it pays (Theorem~2, Corollary~2.1).

\subsection{Remarks}

\paragraph{A.5.1 (Equilibrium selection.)}
The screening game also admits pooling PBE sustained by pessimistic off-path beliefs. When $c > g - c_H$
holds, the separating equilibrium is the unique outcome surviving the Cho--Kreps Intuitive Criterion: from
any pooling candidate, a true $H$ has a deviation---obtain the credential---that is profitable for $H$ but
not for $L$, so the caller should attribute it to $H$, which unravels the pool. This is the standard
Spence/Riley selection \cite{cho1987,riley1979}; we invoke it to justify focusing on separation.

\paragraph{A.5.2 (Reputation as the source of $c$.)}
In a one-shot game, $c$ is the direct screening-plus-penalty cost of obtaining the credential. In the
repeated setting of Section~4.4, being caught confident-wrong forfeits future recognition: with discount
factor $\delta$, per-period premium $g$, and per-period detection probability $p$, the present value an
overclaimer loses upon eventual detection is on the order of $\tfrac{\delta}{1-\delta}\, p\, g$. Reputation
therefore \emph{adds} to the effective $c$, and a sufficiently patient market (high $\delta$) with
non-trivial detection ($p>0$) pushes $c$ above $g$ endogenously. This connects the static condition $c > g$
to the repeated-game mechanism and explains why reputation and screening are complements, not alternatives.

\paragraph{A.5.3 (Continuum of types.)}
Replacing $\{L,H\}$ with a continuum $r \in [r_{\min}, r_{\max}]$ and retaining single-crossing yields the
least-cost (Riley) separating equilibrium, in which the credential intensity rises monotonically with true
reliability and the same incentive logic binds locally. We use two types only for transparency.

\paragraph{A.5.4 (Assumptions and scope.)}
The results require: a screening technology with single-crossing (A4)---i.e., verifiable task instances on
which a true $H$ succeeds more cheaply than an $L$; this fails where capabilities are non-reproducible
(Section~6.5). Reputation's contribution to $c$ (A.5.2) requires persistent identities (Sections~2 and~6.5).
And (A1, A5) assume risk-neutral providers facing a trust price monotone in perceived reliability. Where
these hold, separation is achievable; where they fail---unverifiable tasks, disposable identities---$c$
cannot be raised above $g$ and the lemons problem may be irreducible at the protocol layer, as Section~4.5
notes.

\subsection{The reliability-composition bound (Section 4.6)}

For a delegation chain of depth $n$ with per-hop success probabilities $r_1, \dots, r_n$ and independent
failures, end-to-end success is $\prod_{k=1}^{n} r_k$. If the layer adds a verification step at hop $k$ that
catches a fraction $v_k \in [0,1]$ of that hop's errors, the hop's effective failure probability falls to
$(1 - r_k)(1 - v_k)$, so its effective success probability is
\[
r_k' \;=\; 1 - (1 - r_k)(1 - v_k) \;\ge\; r_k,
\]
and end-to-end success is $\prod_k r_k' \ge \prod_k r_k$: verification helps more as the chain deepens.
Independence is an idealization (Section~4.6). Without it, a conservative planning bound follows from Boole's
inequality---end-to-end failure probability $\le \sum_k (1 - r_k')$---which the layer can compute from
per-hop descriptors without assuming independence. Either way this is a bound and a planning tool, not a
correctness guarantee: by the end-to-end argument, semantic correctness remains with the endpoints.

\end{document}